\documentclass{article}

\usepackage[numbers,sort&compress,longnamesfirst,sectionbib]{natbib}
\usepackage{times}
\usepackage{helvet}
\usepackage{courier}

\usepackage{amsmath, amsthm, amsfonts}
\usepackage{algorithm,algpseudocode}
\usepackage{tablefootnote}
\usepackage{verbatim}
\usepackage{changepage}
\usepackage{multirow}
\usepackage{hhline}
\usepackage{graphicx}
\usepackage{lscape}

\newtheorem{theorem}{Theorem}
\newtheorem{corollary}{Corollary}

\newcommand{\exactMIP}{exactMIP}

\begin{document}
\title{A Fully Polynomial Time Approximation Scheme for Packing While Traveling}
\date{}

\author{
Frank Neumann\textsuperscript{1}, 
Sergey Polyakovskiy\textsuperscript{1}, 
Martin Skutella\textsuperscript{2}, 
Leen Stougie\textsuperscript{3}, \\
Junhua Wu\textsuperscript{1}
\\
\\
\textsuperscript{1} Optimisation and Logistics, School of Computer Science,\\
 The University of Adelaide, Australia
\\
\textsuperscript{2} 
    Combinatorial Optimisation \& Graph Algorithms,\\ Department of Mathematics,\\
 Technical University of Berlin, Germany
\\
\textsuperscript{3} CWI and Operations Research,\\ Dept. of Economics and Business Administration,\\
 Vrije Universiteit, Amsterdam, The Netherlands
}

\maketitle

\begin{abstract}
Understanding the interactions between different combinatorial optimisation problems in real-world applications is a challenging task. Recently, the traveling thief problem (TTP), as a combination of the classical traveling salesperson problem and the knapsack problem, has been introduced to study these interactions in a systematic way. We investigate the underlying non-linear packing while traveling (PWT) problem of the TTP where items have to be selected along a fixed route. We give an exact dynamic programming approach for this problem and a fully polynomial time approximation scheme (FPTAS) when maximising the benefit that can be gained over the baseline travel cost. Our experimental investigations show that our new approaches outperform current state-of-the-art approaches on a wide range of benchmark instances.
\end{abstract}

\section{Introduction}
Combinatorial optimisation problems play a crucial role in important areas such as planning, scheduling and routing. Many combinatorial optimisation problems have been studied extensively in the literature. Two of the most prominent ones are the traveling salesperson problem (TSP) and the knapsack problem (KP) and numerous high performing algorithms have been designed for these two problems. 

Looking at combinatorial optimisation problems arising in real-world applications, one can observe that real-world problems often are composed of different types of combinatorial problems. For example, delivery problems usually consists of a routing part for the vehicle(s) and a packing part of the goods onto the vehicle(s). Recently, the traveling thief problem (TTP)~\cite{Bonyadi2013} has been introduced to study the interactions of different combinatorial optimisation problems in a systematic way and to gain better insights into the design of multi-component problems. The TTP combines the TSP and KP by making the speed that a vehicle travels along a TSP tour dependent on the weight of the selected items. Furthermore, the overall objective is given by the sum of the profits of the collected items minus the weight dependent travel cost along the chosen route. A wide range of heuristic search algorithms~\cite{Faulkner2015,DBLP:journals/soco/MeiLY16,ElYafrani:2016:PVS:2908812.2908847} and a large benchmark set~\cite{Polyakovskiy2014TTP} have been introduced for the TTP in recent years. However, up to now there are no high performing exact approaches to deal with the TTP. 

The study of non-linear planning problems is an important topic and the design of approximation algorithms has gained increasing interest in recent years~\cite{DBLP:conf/aaai/HoyN15,DBLP:conf/aaai/YangN16}. The non-linear packing while traveling problem (PWT) has been introduced in \cite{sergey15} to push forward systematic studies on multi-component problems and deals with the packing part combined with the non-linear travel cost function of the TTP. The PWT can be seen as the TTP when the route is fixed but the cost still depends on the weight of the items on the vehicle. The problem is motivated by gaining advanced precision when minimising transportation costs that may have non-linear nature, for example, in applications where weight impacts the fuel costs~\cite{GOODYEAR,Lin14}. From this point of view, the PWT is a baseline problem in various vehicle routing problems with non-linear costs. Some specific applications of the PWT may deal with a single truck collecting goods in large remote areas without alternative routes, that is a single main route that a vehicle has to follow may exist while any deviations from it in order to visit particular cities are negligible~\cite{DBLP:journals/corr/PolyakovskiyN15}. The problem is $\mathcal{NP}$-hard even without the capacity constraint usually imposed on the knapsack. Furthermore, exact and approximative mixed integer programming approaches as well as a branch-infer-and-bound approach~\cite{DBLP:journals/corr/PolyakovskiyN15} have been developed for this problem.

We introduce a dynamic programming approach for PWT. The key idea is to consider the items in the order they appear on the route that needs to be travelled and apply dynamic programming similar as for the classical knapsack problem~\cite{Toth1980}. When considering an item, the decision has to be made on whether or not to pack the item. The dynamic programming approach computes for the first $i$, $1 \leq i \leq m$, items and each possible weight $w$ the maximal objective value that can be obtained. As the programming table that is used depends on the number of different possible weights, the algorithm runs in pseudo-polynomial time.

After having obtained the exact approach based on dynamic programming, we consider the design of a fully polynomial approximation scheme (FPTAS)~\cite{HochbaumApproximation}. 

First, we show that it is $\mathcal{NP}$-hard to decide whether a given instance of PWT has a non-negative objective value. This rules out any polynomial time algorithm with finite approximation ratio under the assumption $P\not=NP$. Due to this, we design a FPTAS for the amount that can be gained over the travel cost when the vehicle travels empty (which is the minimal possible travel cost). Our FPTAS makes use of the observation that the item with the largest benefit leads to an objective value of at least $OPT/m$ and uses appropriate rounding in the previously designed dynamic programming approach.

We evaluate our two approaches on a wide range of instances from the TTP benchmark set~\cite{Polyakovskiy2014TTP} and compare it to the exact and approximative approaches given in \cite{DBLP:journals/corr/PolyakovskiyN15}. Our results show that the large majority of the instances that can be handled by exact methods, are solved much quicker by dynamic programming than the previously developed mixed integer programming and branch-infer-and-bound approaches. Considering instances with a larger profit and weight range, we show that the choice of the approximation guarantee significantly impacts the runtime behaviour.

The paper is structured as follows. In Section~2, we introduce the problem. We present the exact dynamic programming approach in Section~3 and design a FPTAS in Section~4. Our experimental results are shown in Section~5. Finally, we finish with some conclusions.

\section{Problem Statement}\label{sec:prelim}
The PWT can be formally defined as follows. Given are $n+1$ cities, distances $d_i$, $1\leq i \leq n$, from city $i$ to city $i+1$, and a set of items $M$, $|M| = m$, distributed all over the first $n$ cities. W.l.o.g., we assume $m = \Omega(n)$ to simplify our notations. Each city $i$, $1 \leq i \leq n$, contains a set of items $M_i \subseteq M$, $|M_i| = m_i$. Each item $e_{ij} \in M_i$, $1 \leq j \leq m_i$, is characterised by its positive integer profit $p_{ij}$ and weight $w_{ij}$.

In addition, a fixed route $N = (1, 2, ..., n+1)$ is given that is traveled by a vehicle with velocity $v \in [v_{min},v_{max}]$. Let $x_{ij} \in \{0, 1\}$ be a variable indicating whether or not item $e_{ij}$ is chosen in a solution. Then a set $S \subseteq M$ of selected items can be represented by a decision vector $x = (x_{11},x_{12},...,x_{1m_1},x_{21},...,x_{nm_n}).$ The total benefit of selecting a subset of items $S$ is calculated as

$$	B(x) = P(x) - R \cdot T(x),$$

where

	$$P(x) = \sum\limits_{i=1}^n \sum\limits_{j=1}^{m_i} p_{ij}x_{ij}$$

represents the total profit of selected items and 

$$	T(x) = \sum\limits_{i=1}^n \frac{d_i}{v_{max} - \nu\sum\limits_{k=1}^i\sum\limits_{j=1}^{m_k} w_{kj}x_{kj}}$$
is the total travel time for the vehicle carrying these items.

Here, $\nu = \frac{v_{max}-v_{min}}{W}$ is the constant defined by the input parameters, where $W$ is the capacity of the vehicle. $T(x)$ has the following interpretation: when the vehicle is traveling from city $i$ to city $i+1$, the selected items have to be carried and the maximal speed $v_{max}$ of the vehicle is reduced by a normalised amount that depends linearly on the weight of these items. Because the velocity is influenced by the weight of collected items, the total travel time increases along with their weight. Given a renting rate $R \in (0, \infty)$, $R \cdot T(x)$ is the total cost of carrying the items chosen by $x$. The objective of this problem is to find a solution $x^* = \arg max_{x \in \{0,1\}^m} B(x).$

We investigate dynamic programming and approximation algorithms~\cite{HochbaumApproximation} for the non-linear packing while traveling problem.
A FPTAS for a given maximisation problem is an algorithm $A$ that obtains for any valid input $I$ and $\epsilon$, $0 < \epsilon \leq 1$, a solution of objective value $A(I) \geq (1- \epsilon) OPT(I)$ in time polynomial in the input size $|I|$ and $1/\epsilon$.

\section{Dynamic Programming}\label{sec:dp}

We introduce a dynamic programming approach for solving the PWT. Dynamic programming is one of the traditional approaches for the classical knapsack problem~\cite{Toth1980}. The dynamic programming table $\beta$ consists of $W$ rows and $m$ columns. Items are processed in the order they appear along the path $N$ and we consider them in the lexicographic order with respect to their indices, i.e. 
$$e_{ab} \preceq e_{ij},\text{ iff } ((a < i) \vee ( a=i \wedge  b \leq j)).$$ 
Note that $\preceq$ is a total strict order and we process the items in this order starting with the smallest element. The entry $\beta_{i,j,k}$ represents the maximal benefit that can be obtained by considering all  combinations of items $e_{ab}$ with $e_{ab} \preceq e_{ij}$ leading to weight exactly $k$. We denote by $\beta(i,j, \cdot)$ the column containing the entries $\beta_{i,j,k}$.  In the case that a combination of weight $k$ doesn't exist, we set $\beta_{i,j,k}=-\infty$. We denote by $$d_{in} = \sum_{l=i}^n d_{l}$$ the distance from city $i$ to the last city $n+1$.

We denote by $B(\emptyset)$ the benefit of the empty set which is equivalent to the travel cost when the vehicle travels empty. Furthermore, $B(e_{ij})$ denotes the benefit when only item $e_{ij}$ is chosen.

For the first item $e_{ij}$ according to $\preceq$, we set 
$$\beta(i,j,0)=B(\emptyset),$$ 
$$\beta(i,j,w_{ij})=B(e_{ij}),$$ and $$\beta(i,j,k) = - \infty \text{ iff } k \not \in \{0,w_{ij}\}.$$

Let $e_{i'j'}$ be the predecessor of item $e_{ij}$ in $\preceq$.  Based on $\beta(i',j', \cdot)$ we compute for $\beta(i,j, \cdot)$ each entry $\beta_{i,j,k}$   as

\begin{displaymath}
\max \!\left\{
\begin{array}{l}
\!\!\beta_{i',j',k}
\\
\!\!\beta_{i',j',k-w_{ij}} \!+\! p_{ij} \!-\! Rd_{in} (\frac{1}{v_{max}-\nu k} \!-\! \frac{1}{v_{max}-\nu( k - w_{ij})})\\
\end{array} \right.
\end{displaymath}

Let $e_{st}$ be the last element according to $\preceq$, then $\max_k \beta(s,t,k)$ is reported as the value of an optimal solution.
We now investigate the runtime for this dynamic program.
If $d_{in}$ has been computed for each $i$, $1 \leq i \leq n-1$, which takes $O(n)$ time in total, then each entry can be compute in constant time.

\begin{theorem}
The entry $\beta(i,j,k)$ stores the maximal possible benefit for all subsets of $I_{ij} = \{e_{ab} \mid e_{ab} \preceq e_{ij}\}$ having weight $k$.
\end{theorem}

\begin{proof}
The proof is by induction.
The statement is true for the first item $e_{ij}$ according to $\preceq$ as there are only the two options of choosing or not choosing $e_{ij}$.
Assume that $\beta(i',j',k)$ stores the maximal benefit for each weights $k$ when considering all items of $I_{i'j'}$.
There two options exist when we consider item $e_{ij}$ in addition: to include or not include $e_{ij}$. If $e_{ij}$ is not included, then the best possible value for $\beta(i,j,k)$ is $\beta(i',j',k)$. If $e_{ij}$ is included, then remaining weight has to come from the previous items whose maximal benefit has been $\beta(i',j',k-w_{ij})$. Transporting a set of items of weight $k-w_{ij}$ from city $i$ to city $n+1$ has cost 
$$\frac{Rd_{in}}{v_{max}-\nu ( k - w_{ij})}$$ and transporting a set of items of weight $k$ from city $i$ to $n+1$ has cost
 $$\frac{Rd_{in}}{v_{max}-\nu k}.$$ This cost of transporting items of a fixed weight from city $i$ to city $n+1$ is independent of the choice of items. Therefore, $\beta(i,j,k)$ stores the maximal possible benefit when considering all possible subsets of $I_{ij} = \{e_{ab} \mid e_{ab} \preceq e_{ij}\}$ having weight $k$.
\end{proof}

To speed up the computation of our DP approach, we only store an entry for $\beta(i,j,k)$ if it is not dominated by any other entry in $\beta(i,j, \cdot)$, i.e. there is no other entry $\beta(i,j,k')$ with $\beta(i,j,k') \geq  \beta(i,j,k) \text{ and } k' < k.$ This does not affect the correctness of the approach as an item $e_{ij}$ can be added to any entry of $\beta(i',j', \cdot)$ and therefore we obtain for each dominated entry at least one entry in the last column having at least the same benefit but potentially smaller weight.

\section{Approximation Algorithms}\label{sec:fptas}

We now turn our attention to approximation algorithms.
The NP-hardness proof for PWT given in~\cite{DBLP:journals/corr/PolyakovskiyN15} does not rule out polynomial time approximation algorithms. In this section, we first show that polynomial time approximation algorithms with a finite approximation ratio do not exist under the assumption $P \not = NP$. This motivates the design of a FPTAS for the amount that can be gained over the baseline cost when the vehicle is traveling empty.

\subsection{Inapproximability of PWT}
The objective function for PWT can take on positive and negative values. We show that deciding whether a given PWT instances has a solution that is non-negative is already NP-complete.
\begin{theorem}
\label{thm:nonapprox}
Given a PWT instance,
the problem to decide whether there is a solution $x$ with $B(x) \geq 0$ is NP-complete.
\end{theorem}

\begin{proof}
The problem is in NP as one can verify in polynomial time for a given solution $x$ whether $B(x) \geq 0$ holds by evaluating the objective function. It remains to show that the problem is NP-hard.

We address two cases: when $B(x)$ is subject to the capacity constraint and when it is unconstrained. In both cases, we reduce the $\mathcal{NP}$-complete \textit{subset sum problem} (SSP) to the decision variant of PWT which asks whether there is a solution with objective value at least 0. The input for SSP is given by $m$ positive integers $S=\left\{s_1, \ldots, s_m\right\}$ and a positive integer $Q$. The question is whether there exists a vector $x \in \left\{0,1\right\}^m$ such that $\sum_{k=1}^m {s_kx_k} = Q$. We encode the instance of SSP given by $S$ and $Q$ as the instance of PWT, which consists of two cities. The first city contains all the $m$ items and the distance between the cities is $d_1=1$. We assume that $p_{1k}=w_{1k}=s_k$, $1 \leq k \leq m$.

To prove the first case, we construct the instance $I'$ of PWT. We extend the initial settings by giving to the vehicle capacity $W=Q$ and define its velocity range as $\upsilon_{max}=2$ and $\upsilon_{min}=1$.  Furthermore, we set $R^*=Q$. Consider the nonlinear function $f'_{R^*} \colon \left[0,W\right] \rightarrow \mathbb{R}$ defined as
{\footnotesize
$$ f'_{R^*}\left(w\right)=w-\frac{R^*}{2- w/W}=w-\frac{Q}{2- w/Q}.$$
}
$f'_{R^*}$, which is defined on the interval $\left[0,W\right]$, is a continuous concave function that reaches its unique maximum of 0 in the point $w^* = W = Q$, i.e. $f'_{R^*}\left(w \right)<0$ for $w \in [0,W]$ and $w \not = w^*$. Then 0 is the maximum value for $f'_{R^*}$ when being restricted to integer input, too. Therefore, the objective function for PWT is given by
{\footnotesize
$$ g'_{R^*}\left(x\right)=\displaystyle\sum_{k=1}^m {p_{1k}x_k}-\frac{R^*}{2- \frac{1}{W} \displaystyle\sum_{k=1}^m {w_{1k}x_k}}.$$
}
There exists an $x \in \{0,1\}^m$ such that $g'_{R^*}(x) \geq 0$ iff $$\sum_{k=1}^m s_kx_k=\sum_{k=1}^m w_{1k}x_k=\sum_{k=1}^m p_{1k}x_k=Q.$$ Therefore, the instance of SSP has answer YES iff the optimal solution of the PWT instance $I'$ has objective value at least 0. Obviously, the reduction can be carried out in polynomial time which completes the proof of the first case.

To prove the second case, we construct the instance $I''$ of PWT where our settings assume $$W = \sum_{k=1}^m s_k$$ and $$\upsilon_{min}=\sqrt{Q/(2W-Q)} = \upsilon_{max}/2.$$ We then set $$R^* =  \upsilon_{min} \cdot W \left(\upsilon_{max} - \upsilon_{min} \cdot Q/W\right)^2.$$ Finally, this gives us the functions $f''_{R^*}\left(w\right)$ and $g''_{R^*}\left(x\right)$ of the following forms:
{\footnotesize
$$ f''_{R^*}\left(w\right)=w-\frac{R^*}{\upsilon_{max}- \upsilon_{min} \cdot w/W}.$$
}
{\footnotesize
$$ g''_{R^*}\left(x\right)=\displaystyle\sum_{k=1}^m {p_{1k}x_k}-\frac{R^*}{\upsilon_{max}- \frac{\upsilon_{min}}{W} \displaystyle\sum_{k=1}^m {w_{1k}x_k}}.$$
}

Similarly, there exists an $x \in \{0,1\}^m$ such that $g''_{R^*}(x) \geq 0$ iff $$\sum_{k=1}^m s_kx_k=\sum_{k=1}^m w_{1k}x_k=\sum_{k=1}^m p_{1k}x_k=Q.$$ Therefore, the instance of SSP has answer YES iff the optimal solution of the PWT instance $I''$ has objective value at least 0, while the reduction can be carried out in polynomial time.

\end{proof}
The objective function can take on negative and non-negative values.
Theorem~\ref{thm:nonapprox} rules out meaningful approximations for the original objective functions $B$ and we state this in the following corollary.

\begin{corollary}
There is no polynomial time approximation algorithm for PWT with a meaningful approximation ratio, unless P=NP.
\end{corollary}

\subsection{FPTAS for amount over baseline travel cost}
As there are no polynomial time approximation algorithms for fixed approximation ratio for PWT, we consider the amount that can be gained over the cost when the vehicle travels empty as the objective. This is motivated by the scenario where the vehicle has to travel along the given route and the goal is to maximise the gain over this baseline cost. Note that an optimal solution for this objective is also an optimal solution for PWT. However, approximation results do not carry over to PWT as the objective values are ``shifted'' by the cost when traveling empty.

\begin{algorithm*}[th]
\begin{itemize}
\item Set $L = \max_{e_{ij} \in M} B'(e_{ij})$, $r = \epsilon L /m$, and $d_{in} = \sum_{l=i}^n d_{l}$, $1\leq i \leq n$.
\item Compute order $\preceq$ on the items $e_{ij}$ by sorting them in lexicographic order with respect to their indices $(i,j)$.
\item For the first item $e_{ij}$ according to $\preceq$, set $\beta(i,j,0) = B'(\emptyset)$ and $\beta(i,j,w_{ij}) = B'(e_{ij})$.
\item Consider the remaining items of $M$ in the order of $\preceq$ and do for each item  $e_{ij}$ and its predecessor $e_{i'j'}$:
\begin{itemize}

\item In increasing order of $k$ do for each $\beta(i',j',k)$ with $\beta(i',j',k)\not = -\infty$
\begin{itemize}
\item If there is no $\beta(i,j,k')$ with ($\lfloor \beta(i,j,k')/r\rfloor \geq \lfloor \beta(i',j',k)/r\rfloor$ and $k'<k$),\\ set  $\beta(i,j,k) = max\{ \beta(i,j,k), \beta(i',j',k) \}$.
\item If there is no $\beta(i,j,k')$ with ($\lfloor \beta(i,j,k')/r\rfloor \geq \lfloor \beta(i',j',k+w_{ij})/r\rfloor$ and $k'<k+w_{ij}$),\\ set $\beta(i,j,k+w_{ij}) = max\{ \beta(i,j,k+w_{ij}), \beta(i',j',k) +  p_{ij} + Rd_{in} ( \frac{1}{v_{\max}-\nu k} - \frac{1}{v_{max}-\nu (k+w_{ij})})  \}$.

\end{itemize}

\end{itemize}
\end{itemize}
\caption{FPTAS for $B'(x)$}
\label{alg:fptas}
\end{algorithm*}

Let 
$$B(\emptyset) = - R \cdot \sum_{i=1}^n d_i / v_{\max}$$
be the travel cost (or benefit) for the empty truck. $B(\emptyset)$ can be seen as the set up cost that we have to pay at least.
We consider the objective  
$$B'(x) = B(x) - B(\emptyset),$$
 i. e. for the amount that we can gain over this setup cost, and give an FPTAS.
Note, that we have $-R \cdot T(x) \leq B(\emptyset)$ for any $x \in \{0,1\}^m$ and $P(x)- R \cdot T(x) - B(\emptyset)=0$ if $x = 0^m$. 

We now give a FPTAS for the amount that can be gained over the cost when the vehicle travels empty
and denote by OPT the optimal value for this objective, i.e.
$$
OPT = \max_{x \in \{0,1\}^m} B'(x).
$$

Considering the dynamic program for $B'(x)$ instead of $B(x)$ increases each entry by $|B(\emptyset)|$ and therefore obtains an optimal solution for $B'(x)$ in pseudo-polynomial time. In order to obtain an FPTAS, we round the values of $B'(x)$ and store for each rounded value only the minimal achievable weight.

Let 
$$
t(w)= \frac{1}{v_{\max} - \nu w}
$$
denote the travel time per unit distance when traveling with weight $w$.
We have $t(x+w) - t(x) \geq t(w)$ for any $x\geq 0$ as $t(w)$ is a convex function.

Consider the value $B(e_{ij}) - B(\emptyset)$ which gives the additional amount over $B(\emptyset)$  when only packing item $e_{ij}$.
We assume that there exists at least one item $e_{ij}$ with  $B(e_{ij}) - B(\emptyset) >0$ as otherwise $OPT=0$ the solution being $\{0\}^m$.
 Let $P(e_{ij})$ and $T(e_{ij})$ be the profit and travel time when only choosing item $e_{ij}$. Furthermore, let $x^* = \arg \max_{x \in \{0,1\}^m} B'(x)$ be an optimal solution of value $OPT>0$.

We have 
$$
 \sum_{i=1}^n \sum_{j=1}^{m_i} (P(e_{ij})- R \cdot T(e_{ij})) x_{ij}^* -  B(\emptyset)
  \geq   B(x^*) - B(\emptyset) = OPT
$$
as $t(w)$ is monotonically increasing and convex.

 Therefore the item $e_{ij}$ of $x^*$ with $B(e_{ij}) - B(\emptyset)>0$ maximal fulfils  
$B(e_{ij}) - B(\emptyset) \geq OPT/m.
$

Let $$L =max_{e_{ij} \in M}  B'(e_{ij}) >0$$ be  maximal possible objective value when choosing exactly one item. We have $$L \geq OPT/m \text{ and }L \leq OPT.$$

We set $r = \epsilon L /m$, where $\epsilon$ is the approximation parameter for the FPTAS.
For the FPTAS we round $B'(x)$ to $\lfloor(B'(x)/r \rfloor$ and store for each of such values the minimal weight obtained.
As we only store entries with  $0 \leq B'(x) \leq OPT$, and for each such integer based on dominance and rounding one entry, the total number of entries per column is upper bounded by 
$$(OPT/ r) +1 \leq OPT/ (\epsilon L/m) + 1 \leq m^2/\epsilon +1$$
and number of entries in the dynamic programming table is $O(m^3 /\epsilon)$.

In each step, we make an error of at most $$r = \epsilon L /m \leq \epsilon OPT /m$$ and the error after $m$ steps is at most $\epsilon L \leq \epsilon OPT.$ Hence, the solution $x$ with maximal $B'$-value after having considered all items fulfils $$B'(x) \geq (1-\epsilon) OPT.$$

To implement the idea (see Algorithm~\ref{alg:fptas}), we only store an entry $\beta(i,j,k)$ if there is no entry $\beta(i,j,k')$ with $$\lfloor \beta(i,j,k')/r\rfloor \geq \lfloor \beta(i,j,k)/r\rfloor \text{ and }k'<k.$$ 
Hence, for each possible value $\lfloor \beta(i,j,k)/r\rfloor$ at most one entry is stored and the number of entries for each column $\beta(i,j,\cdot)$ is upper bounded by  $m^2/\epsilon +1$ (as stated above). 
Using for each $\beta(i,j,\cdot)$ a list which stores the entries $\beta(i,j,k)$ in increasing order of $k$ can be used for our implementation.

Based on our investigations and the design of Algorithm~\ref{alg:fptas}, we can state the following result.
\begin{theorem}
Algorithm~\ref{alg:fptas} is a fully polynomial time approximation scheme (FPTAS) for the objective $B'$. It obtains for any $\epsilon$, $0< \epsilon \leq 1$, a solution $x$ with $B'(x) \geq (1-\epsilon) \cdot OPT$ in time $O(m^3/\epsilon)$.
\end{theorem}

The construction of the FPTAS only used the fact that the travel time per unit distance is monotonically increasing and convex. Hence, the FPTAS holds for any PWT problem where the travel time per unit distance has this property.

\section{Experiments and Results}\label{sec:experiments}

In this section, we investigate the effectiveness of the proposed DP and FPTAS approaches based on our implementations in Java\footnote{The code will be made available online at time of publication.}. We mainly focus on two issues: 1) studying how the DP and FPTAS perform compared to the state-of-the-art approaches; 2) investigating how the performance and accuracy of the FPTAS change when the parameter $\epsilon$ is altered.

\begin{landscape}
\setlength{\tabcolsep}{3pt}
\begin{table*}
\centering
\caption{Results on  Small Range Instances}
\label{tab:smallresfptas}
{\footnotesize
\begin{adjustwidth}{0cm}{}
\begin{tabular}{rc|@{\,}r@{\,}|r|r|r|rr|rr|rr|rr|rr|rr}
\hline
\multirow{4}{*}{Instance}&	\multirow{4}{*}{m}& \multicolumn{1}{c|}{\multirow{4}{*}{OPT}} & \multicolumn{3}{c|}{Exact Approaches} & \multicolumn{12}{c}{Approximation Approaches} \\
\hhline{~~~---------------}
&&&	\multicolumn{1}{c|}{\exactMIP} &\multicolumn{1}{c|}{BIB} &\multicolumn{1}{c|}{DP} &\multicolumn{2}{c|}{approxMIP} &\multicolumn{10}{c}{FPTAS} \\
\hhline{~~~---------------}
&&&&&&&&	\multicolumn{2}{c|}{$\epsilon=0.0001$}&	\multicolumn{2}{c|}{$\epsilon=0.01$}&	\multicolumn{2}{c|}{$\epsilon=0.1$}&	\multicolumn{2}{c|}{$\epsilon=0.25$}&	\multicolumn{2}{c}{$\epsilon=0.75$}\\
&	&	& 
RT(s)& RT(s)& RT(s)&
AR(\%)& RT(s)&
AR(\%)& RT(s)&
AR(\%)& RT(s)&
AR(\%)& RT(s)&
AR(\%)& RT(s)&
AR(\%)& RT(s)\\
\hline																																					
\multicolumn{18}{c}{instance family \texttt{eil101}} 	\\																																		
\hline																																			
uncorr\_01	&	100	&	1651.6970	&	1.217	&	5.694	&	0.027	&	100.0000	&	3.838	&	100.0000	&	0.001	&	100.0000	&	0.001	&	100.0000	&	0.001	&	100.0000	&	0.001	&	100.0000	&	0.025	\\
uncorr\_06	&	100	&	10155.4942	&	12.605	&	3.698	&	0.065	&	100.0000	&	4.961	&	100.0000	&	0.012	&	100.0000	&	0.011	&	100.0000	&	0.011	&	100.0000	&	0.011	&	99.9928	&	0.063	\\
uncorr\_10	&	100	&	10297.7134	&	3.525	&	0.795	&	0.036	&	100.0000	&	0.624	&	100.0000	&	0.017	&	100.0000	&	0.017	&	99.9939	&	0.016	&	99.9939	&	0.016	&	99.9653	&	0.037	\\
uncorr-s-w\_01	&	100	&	2152.6188	&	0.328	&	7.566	&	0.001	&	100.0000	&	3.978	&	100.0000	&	0.000	&	100.0000	&	0.000	&	100.0000	&	0.000	&	100.0000	&	0.000	&	100.0000	&	0.003	\\
uncorr-s-w\_06	&	100	&	4333.8512	&	12.590	&	2.215	&	0.012	&	100.0000	&	2.699	&	100.0000	&	0.008	&	100.0000	&	0.007	&	100.0000	&	0.007	&	99.9569	&	0.008	&	99.9569	&	0.017	\\
uncorr-s-w\_10	&	100	&	9048.4908	&	37.144	&	1.107	&	0.022	&	100.0000	&	1.763	&	100.0000	&	0.012	&	100.0000	&	0.012	&	100.0000	&	0.012	&	100.0000	&	0.013	&	99.9355	&	0.020	\\
b-s-corr\_01	&	100	&	4441.9852	&	1.420	&	125.954	&	0.014	&	100.0000	&	5.366	&	100.0000	&	0.010	&	100.0000	&	0.009	&	100.0000	&	0.009	&	100.0000	&	0.008	&	100.0000	&	0.013	\\
b-s-corr\_06	&	100	&	10260.9767	&	4.509	&	22.541	&	0.101	&	100.0000	&	2.761	&	100.0000	&	0.058	&	100.0000	&	0.057	&	100.0000	&	0.048	&	100.0000	&	0.043	&	100.0000	&	0.087	\\
b-s-corr\_10	&	100	&	13630.6153	&	11.013	&	27.081	&	0.187	&	99.9971	&	3.713	&	100.0000	&	0.103	&	100.0000	&	0.101	&	99.9971	&	0.081	&	99.9606	&	0.065	&	99.8143	&	0.113	\\
uncorr\_01	&	500	&	17608.5781	&	19.594	&	27.581	&	0.247	&	100.0000	&	5.757	&	100.0000	&	0.171	&	100.0000	&	0.161	&	100.0000	&	0.153	&	100.0000	&	0.163	&	100.0000	&	0.377	\\
uncorr\_06	&	500	&	56294.5239	&	384.213	&	13.354	&	2.829	&	100.0000	&	7.800	&	100.0000	&	2.370	&	100.0000	&	2.344	&	100.0000	&	2.300	&	100.0000	&	2.212	&	100.0000	&	2.340	\\
uncorr\_10	&	500	&	66141.4840	&	211.302	&	2.325	&	4.010	&	100.0000	&	0.718	&	100.0000	&	3.720	&	100.0000	&	3.645	&	100.0000	&	3.446	&	100.0000	&	3.531	&	100.0000	&	3.632	\\
uncorr-s-w\_01	&	500	&	13418.8406	&	4.337	&	34.866	&	0.090	&	100.0000	&	50.310	&	100.0000	&	0.085	&	100.0000	&	0.090	&	100.0000	&	0.084	&	100.0000	&	0.087	&	99.9910	&	0.085	\\
uncorr-s-w\_06	&	500	&	34280.4730	&	346.430	&	7.285	&	1.040	&	100.0000	&	9.609	&	100.0000	&	0.964	&	100.0000	&	0.933	&	100.0000	&	0.905	&	100.0000	&	0.936	&	100.0000	&	0.920	\\
uncorr-s-w\_10	&	500	&	50836.6588	&	519.902	&	3.338	&	2.022	&	100.0000	&	3.354	&	100.0000	&	2.005	&	100.0000	&	1.783	&	100.0000	&	1.753	&	100.0000	&	1.784	&	100.0000	&	2.147	\\
b-s-corr\_01	&	500	&	21306.9158	&	40.482	&	624.204	&	1.534	&	100.0000	&	13.338	&	100.0000	&	1.373	&	100.0000	&	1.279	&	100.0000	&	1.116	&	100.0000	&	0.949	&	100.0000	&	0.716	\\
b-s-corr\_06	&	500	&	69370.2367	&	236.387	&	97.313	&	14.616	&	99.9996	&	7.847	&	100.0000	&	13.393	&	100.0000	&	12.975	&	100.0000	&	11.642	&	99.9996	&	9.741	&	99.9996	&	6.018	\\
b-s-corr\_10	&	500	&	82033.9452	&	376.569	&	218.728	&	22.011	&	100.0000	&	2.309	&	100.0000	&	21.372	&	100.0000	&	20.829	&	100.0000	&	18.573	&	100.0000	&	15.313	&	99.9943	&	8.840	\\
uncorr\_01	&	1000	&	36170.9109	&	218.306	&	114.567	&	1.872	&	99.9993	&	11.918	&	100.0000	&	1.891	&	100.0000	&	1.875	&	100.0000	&	1.832	&	100.0000	&	1.845	&	100.0000	&	1.764	\\
uncorr\_06	&	1000	&	93949.1981	&	1261.949	&	36.847	&	20.944	&	100.0000	&	17.971	&	100.0000	&	17.024	&	100.0000	&	16.615	&	100.0000	&	16.545	&	100.0000	&	16.378	&	100.0000	&	15.713	\\
uncorr\_10	&	1000	&	122963.6617	&	620.896	&	4.821	&	30.116	&	100.0000	&	2.184	&	100.0000	&	27.305	&	100.0000	&	26.783	&	100.0000	&	26.541	&	100.0000	&	26.051	&	100.0000	&	23.905	\\
uncorr-s-w\_01	&	1000	&	27800.9614	&	241.957	&	399.158	&	0.802	&	100.0000	&	4985.566	&	100.0000	&	0.730	&	100.0000	&	0.690	&	100.0000	&	0.688	&	100.0000	&	0.724	&	100.0000	&	0.687	\\
uncorr-s-w\_06	&	1000	&	61764.4599	&	1152.624	&	12.792	&	9.872	&	100.0000	&	19.063	&	100.0000	&	8.686	&	100.0000	&	8.812	&	100.0000	&	8.560	&	100.0000	&	8.740	&	100.0000	&	8.396	\\
uncorr-s-w\_10	&	1000	&	103572.4074	&	2146.408	&	7.644	&	15.047	&	100.0000	&	9.688	&	100.0000	&	14.030	&	100.0000	&	13.912	&	100.0000	&	13.797	&	100.0000	&	13.982	&	100.0000	&	13.492	\\
b-s-corr\_01	&	1000	&	46886.1094	&	378.551	&	6129.531	&	11.783	&	99.9988	&	46.394	&	100.0000	&	11.714	&	100.0000	&	11.358	&	100.0000	&	10.793	&	100.0000	&	9.592	&	100.0000	&	6.536	\\
b-s-corr\_06	&	1000	&	125830.6887	&	643.533	&	919.201	&	94.523	&	99.9999	&	10.311	&	100.0000	&	92.411	&	100.0000	&	91.039	&	100.0000	&	83.002	&	99.9999	&	71.078	&	100.0000	&	45.433	\\
b-s-corr\_10	&	1000	&	161990.5015	&	862.572	&	1646.520	&	151.601	&	100.0000	&	7.160	&	100.0000	&	150.279	&	100.0000	&	149.722	&	100.0000	&	134.764	&	100.0000	&	113.049	&	99.9981	&	70.135	\\
\hline																																					
\end{tabular}
\end{adjustwidth}
}

\end{table*}
\end{landscape}

\begin{landscape}
\setlength{\tabcolsep}{2pt}
\begin{table*}
\centering
\caption{Results of DP and FPTAS on Large Range Instances}
\label{tab:mediumresfptas}
{\footnotesize
\begin{adjustwidth}{0cm}{}
\begin{tabular}{rc|rr|rr|rr|rr|rr|rr|rr|rr}
\hline
\multirow{3}{*}{Instance}&	\multirow{3}{*}{m}& \multicolumn{2}{c|}{DP} &\multicolumn{14}{c}{FPTAS} \\
\hhline{~~----------------}
&&&&	\multicolumn{2}{c|}{$\epsilon=0.0001$}&	\multicolumn{2}{c|}{$\epsilon=0.001$}&	\multicolumn{2}{c|}{$\epsilon=0.01$}&	\multicolumn{2}{c|}{$\epsilon=0.1$}&	\multicolumn{2}{c|}{$\epsilon=0.25$}&	\multicolumn{2}{c|}{$\epsilon=0.5$}&	\multicolumn{2}{c}{$\epsilon=0.75$}\\
&	& 
OPT& RT(s)&
AR(\%)& RT(s)&
AR(\%)& RT(s)&
AR(\%)& RT(s)&
AR(\%)& RT(s)&
AR(\%)& RT(s)&
AR(\%)& RT(s)&
AR(\%)&RT(s)\\
\hline \multicolumn{18}{c}{instance family \texttt{eil101\_large-range}} \\ \hline																																	
uncorr\_01&	100&	69802802.2801	&	0.030	&	100.0000	&	0.002	&	100.0000	&	0.002	&	100.0000	&	0.002	&	100.0000	&	0.002	&	100.0000	&	0.002	&	100.0000	&	0.002	&	100.0000	&	0.029	\\
uncorr\_06&	100&	204813765.6933	&	0.053	&	100.0000	&	0.019	&	100.0000	&	0.020	&	100.0000	&	0.019	&	100.0000	&	0.019	&	100.0000	&	0.019	&	100.0000	&	0.019	&	100.0000	&	0.049	\\
uncorr\_10&	100&	172176182.1249	&	0.041	&	100.0000	&	0.028	&	100.0000	&	0.028	&	100.0000	&	0.028	&	100.0000	&	0.028	&	100.0000	&	0.027	&	100.0000	&	0.026	&	99.9628	&	0.037	\\
uncorr-s-w\_01&	100&	36420530.5753	&	0.006	&	100.0000	&	0.003	&	100.0000	&	0.003	&	100.0000	&	0.003	&	100.0000	&	0.003	&	100.0000	&	0.003	&	100.0000	&	0.002	&	100.0000	&	0.004	\\
uncorr-s-w\_06&	100&	148058928.2952	&	0.098	&	100.0000	&	0.072	&	100.0000	&	0.502	&	100.0000	&	0.072	&	100.0000	&	0.069	&	100.0000	&	0.065	&	100.0000	&	0.059	&	100.0000	&	0.070	\\
uncorr-s-w\_10&	100&	142538516.4602	&	0.136	&	100.0000	&	0.101	&	100.0000	&	0.104	&	100.0000	&	0.103	&	99.9978	&	0.096	&	99.9978	&	0.086	&	99.9978	&	0.073	&	99.9978	&	0.089	\\
m-s-corr\_01&	100&	19549602.2671	&	0.003	&	100.0000	&	0.002	&	100.0000	&	0.002	&	100.0000	&	0.002	&	100.0000	&	0.002	&	100.0000	&	0.002	&	100.0000	&	0.001	&	100.0000	&	0.002	\\
m-s-corr\_06&	100&	137203175.1921	&	0.147	&	100.0000	&	0.115	&	100.0000	&	0.118	&	100.0000	&	0.113	&	100.0000	&	0.089	&	100.0000	&	0.063	&	100.0000	&	0.040	&	100.0000	&	0.043	\\
m-s-corr\_10&	100&	225584278.6004	&	0.424	&	100.0000	&	0.326	&	100.0000	&	0.329	&	100.0000	&	0.312	&	100.0000	&	0.200	&	100.0000	&	0.179	&	100.0000	&	0.086	&	100.0000	&	0.073	\\
uncorr\_01&	500&	385692662.0930	&	0.470	&	100.0000	&	0.451	&	100.0000	&	0.454	&	100.0000	&	0.619	&	100.0000	&	0.508	&	100.0000	&	0.445	&	100.0000	&	0.430	&	100.0000	&	0.517	\\
uncorr\_06&	500&	958013934.6172	&	3.539	&	100.0000	&	3.749	&	100.0000	&	7.431	&	100.0000	&	3.947	&	100.0000	&	3.690	&	99.9996	&	3.677	&	99.9996	&	3.486	&	99.9993	&	3.021	\\
uncorr\_10&	500&	844949838.4389	&	4.870	&	100.0000	&	5.393	&	100.0000	&	5.716	&	100.0000	&	5.483	&	100.0000	&	5.135	&	100.0000	&	4.851	&	99.9992	&	4.609	&	99.9992	&	4.295	\\
uncorr-s-w\_01&	500&	182418888.9364	&	1.157	&	100.0000	&	1.157	&	100.0000	&	1.199	&	100.0000	&	1.145	&	99.9995	&	1.112	&	99.9995	&	1.063	&	99.9995	&	0.977	&	99.9904	&	0.929	\\
uncorr-s-w\_06&	500&	780432253.0187	&	22.390	&	100.0000	&	25.040	&	100.0000	&	26.276	&	100.0000	&	24.024	&	100.0000	&	23.282	&	99.9997	&	21.756	&	99.9997	&	18.293	&	99.9997	&	18.411	\\
uncorr-s-w\_10&	500&	714433353.7957	&	30.959	&	100.0000	&	34.458	&	100.0000	&	39.004	&	100.0000	&	34.308	&	100.0000	&	32.308	&	99.9996	&	28.792	&	99.9990	&	26.392	&	99.9990	&	25.971	\\
m-s-corr\_01&	500&	96463941.1275	&	2.335	&	100.0000	&	2.478	&	100.0000	&	2.782	&	100.0000	&	2.695	&	100.0000	&	1.509	&	100.0000	&	0.963	&	100.0000	&	0.546	&	100.0000	&	0.408	\\
m-s-corr\_06&	500&	666701000.1488	&	108.705	&	100.0000	&	126.833	&	100.0000	&	139.630	&	100.0000	&	122.750	&	100.0000	&	62.479	&	100.0000	&	33.547	&	100.0000	&	17.959	&	100.0000	&	10.642	\\
m-s-corr\_10&	500&	1082009880.5886	&	262.999	&	100.0000	&	299.862	&	100.0000	&	317.352	&	100.0000	&	274.284	&	100.0000	&	145.087	&	100.0000	&	78.470	&	99.9994	&	41.816	&	99.9994	&	25.924	\\
uncorr\_01&	1000&	777386336.9660	&	4.222	&	100.0000	&	4.397	&	100.0000	&	4.347	&	100.0000	&	4.309	&	100.0000	&	4.341	&	100.0000	&	4.377	&	100.0000	&	4.280	&	100.0000	&	4.240	\\
uncorr\_06&	1000&	1933319297.4248	&	46.043	&	100.0000	&	51.383	&	100.0000	&	53.087	&	100.0000	&	48.861	&	100.0000	&	52.957	&	99.9999	&	52.062	&	99.9997	&	50.286	&	99.9996	&	51.488	\\
uncorr\_10&	1000&	1693797490.1704	&	64.485	&	100.0000	&	76.744	&	100.0000	&	78.847	&	100.0000	&	74.128	&	100.0000	&	82.754	&	100.0000	&	77.057	&	100.0000	&	72.283	&	100.0000	&	72.567	\\
uncorr-s-w\_01&	1000&	361991311.8336	&	14.254	&	100.0000	&	15.072	&	100.0000	&	15.670	&	100.0000	&	14.523	&	100.0000	&	14.110	&	100.0000	&	14.039	&	100.0000	&	12.088	&	100.0000	&	11.129	\\
uncorr-s-w\_06&	1000&	1574469459.3163	&	286.843	&	100.0000	&	318.096	&	100.0000	&	330.508	&	100.0000	&	337.289	&	100.0000	&	334.318	&	100.0000	&	307.588	&	99.9998	&	270.013	&	99.9996	&	245.927	\\
uncorr-s-w\_10&	1000&	1439410696.3695	&	393.793	&	100.0000	&	438.775	&	100.0000	&	455.830	&	100.0000	&	464.527	&	100.0000	&	441.955	&	100.0000	&	433.672	&	99.9994	&	378.917	&	99.9994	&	340.813	\\
m-s-corr\_01&	1000&	191170309.5684	&	46.858	&	100.0000	&	58.031	&	100.0000	&	59.987	&	100.0000	&	58.101	&	100.0000	&	31.703	&	100.0000	&	18.771	&	100.0000	&	10.728	&	100.0000	&	6.831	\\
m-s-corr\_06&	1000&	1315708161.7720	&	2393.205	&	100.0000	&	2512.281	&	100.0000	&	2606.412	&	100.0000	&	1921.573	&	100.0000	&	666.749	&	100.0000	&	364.452	&	100.0000	&	208.969	&	100.0000	&	150.060	\\
m-s-corr\_10&	1000&	2163713055.3759	&	6761.490	&	100.0000	&	6668.535	&	100.0000	&	6441.906	&	100.0000	&	4526.653	&	100.0000	&	1334.882	&	100.0000	&	703.258	&	100.0000	&	397.527	&	100.0000	&	282.211	\\
\hline																																					
\end{tabular}
\end{adjustwidth}
}
\end{table*}
\end{landscape}

In order to be comparable to the mixed integer programming (MIP) and the branch-infer-and-bound (BIB) approaches presented in ~\cite{DBLP:journals/corr/PolyakovskiyN15}, we conduct our experiments on the same families of test instances. Our experiments are carried out on a computer with 4GB RAM and a 3.06GHz Intel Dual Core processor, which is also the same as the machine used in the paper mentioned above. 

We compare the DP to the exact MIP (\textit{{\exactMIP}}) and the branch-infer-and-bound approaches as well as the FPTAS to the approximate MIP (\textit{approxMIP}), as the former three are all exact approaches and the latter two are all approximations. Table~\ref{tab:smallresfptas} demonstrates the results for a route of 101 cities and various types of packing instances. For this particular family, we consider three types of instances: \textit{uncorrelated} (uncorr), \textit{uncorrelated with similar weights} (uncorr-s-w) and \textit{bounded strongly correlated} (b-s-corr), which are further distinguished by the different correlations between profits and weights. In combination with three different numbers of items and three settings of the capacity, we have 27 instances in total, as shown in the column called ``\textit{Instance}''. Similarly to the settings in~\cite{DBLP:journals/corr/PolyakovskiyN15}, every instance with ``\_01'' postfix has a relatively small capacity. We expect such instances to be potentially easy to solve by DP and FPTAS due to the nature of the algorithms. The \textit{OPT} column shows the optimum of each instance and the \textit{RT(s)} columns illustrate the running time for each of the approaches in the time unit of a second. To demonstrate the quality of an approximate approach applied to the instances, we use the ratio between the objective value obtained by the algorithm and the optimum obtained for an instance as the approximation rate $AR(\%) = 100 \times \frac{OBJ}{OPT}$.

In the comparison of exact approaches, our results show that the DP is much quicker than the exact MIP and BIB in solving the majority of the instances. The exact MIP is slower than the DP in every case and this dominance is mostly significant. For example, it spends around $35$ minutes to solve the instance \textit{uncorr-s-w\_10} with $1,000$ items, where the DP needs around $15$ seconds only. On the other hand, the BIB slightly beats the DP on three instances, but the DP is superior for the rest $24$ instances. An extreme case is \textit{b-s-corr\_01} with $1,000$ items where the BIB spends above $1.5$ hours while the DP solves it in $11$ seconds only. Concerning the running time of the DP, it significantly increases only for the instances having large amount of items with strongly correlated weights and profits, such as \textit{b-s-corr\_06} and \textit{b-s-corr\_10} with $1,000$ items. However, \textit{b-s-corr\_01} seems exceptional due to the limited capacity assigned to the instance.  

Our comparison between the approximation approaches shows that the FPTAS has significant advantages as well. The approximation ratios remain $100\%$ when $\epsilon$ equals $0.0001$ and $0.01$. Only when $\epsilon$ is set to $0.25$, the FPTAS starts to output the results having similar accuracies as the ones of \textit{approxMIP}. With regard to the performance, the FPTAS takes less running time than \textit{approxMIP} on the majority of the instances despite the setting of $\epsilon$. As an extreme case, \textit{approxMIP} requires hours to solve the \textit{uncorr-s-w\_01} instance with $1,000$ items, but the FPTAS takes less than a second. However, the \textit{approxMIP} performs much better on \textit{b-s-corr\_06} and \textit{b-s-corr\_10} with $1,000$ items. This somehow indicates that the underlying factors that make instances hard to solve by approximate MIP and FPTAS have different nature. Understanding these factors more and using them wisely should help to build a more powerful algorithm with mixed features of MIP and FPTAS.

In our second experiment, we use test instances which are slightly different to those in the benchmark used in ~\cite{DBLP:journals/corr/PolyakovskiyN15}. This is motivated by our findings that relaxing $\epsilon$ from $0.0001$ to $0.75$ improves the performance of FPTAS by around $50\%$ for the b-s-corr instances, while does not degrade the accuracy noticeably. At the same time, there is no significant improvement for other instances. It's surprising as shows that the performance improvement can be easily achieved on complex instances. Therefore, we study how the FPTAS performs if the instances are more complicated. The idea is to use instances with large weights, which are known to be difficult regarding dynamic programming based approaches for the classical knapsack problem. We follow the same way to create TTP instances as proposed in ~\cite{Polyakovskiy2014TTP} and generate the knapsack component of the problem as discussed in ~\cite{Pisinger20052271}. Specifically, we extend the range to generate potential profits and weights from $[1, 10^3]$ to $[1, 10^7]$ and focus on \textit{uncorrelated} (uncorr), \textit{uncorrelated with similar weights} (uncorr-s-w), and \textit{multiple strongly correlated} (m-s-corr) types of instances. Additionally, in the stage of assigning the items of a knapsack instance to particular cities of a given TSP tour, we sort the items in descending order of their profits and the second city obtains $k$, $k\in\left\{1,5,10\right\}$, items of the largest profits, the third city then has the next $k$ items, and so on. We expect that such assignment should force the algorithms to select items in the first cities of a route making the instances more challenging for the DP and FPTAS. In fact, these instances occur to be harder and force us to switch to the 128GB RAM and 8 $\times$~(2.8GHz AMD 6 core processors) cluster machine to carry out the second experiment.

Table~\ref{tab:mediumresfptas} illustrates the results of running the DP and FPTAS on the instances with the large range of profits and weights. Generally speaking, we can observe that the instances are significantly harder to solve than those ones from the first experiment, that is they take comparably more time. Similarly, the instances with large number of items, larger capacity, and strong correlation between profits and weights are now hard for the DP as well. Oppositely to the results of the previous experiment, the FPTAS performs much better when dealing with such instances in the case when $\epsilon$ is relaxed. For example, its performance is improved by $95\%$ for the instance \textit{m-s-corr\_10} with $1,000$ items when $\epsilon$ is raised from $0.0001$ to $0.75$ while the approximation rate remains at $100\%$.

\section{Conclusion}\label{sec:conclusion}
Multi-component combinatorial optimisation problems play an important role in many real-world applications. We have examined the non-linear packing while traveling problem which results from the interactions in the TTP. We designed a dynamic programming algorithm that solves the problem in pseudo-polynomial time. Furthermore, we have shown that the original objective of the problem is hard to approximate and have given an FPTAS for optimising the amount that can be gained over the smallest possible travel cost. It should be noted that the FPTAS applies to a wider range of problems as our proof only assumed that the travel cost per unit distance in dependence of with weight $w$ is monotone increasing and convex. Our experimental results on different types of knapsack instances show the advantage of the dynamic program over the previous approach based on mixed integer programming and branch-infer-and-bound concepts. Furthermore, we have demonstrated the effectiveness of the FPTAS on instances with a large weight and profit range.

\section*{Acknowledgements}

The authors were supported by Australian Research Council grants DP130104395 and DP140103400.

\bibliographystyle{abbrv}
\bibliography{references} 

\begin{thebibliography}{}

\bibitem[\protect\citeauthoryear{Bonyadi, Michalewicz, and
  Barone}{2013}]{Bonyadi2013}
Bonyadi, M.; Michalewicz, Z.; and Barone, L.
\newblock 2013.
\newblock The travelling thief problem: The first step in the transition from
  theoretical problems to realistic problems.
\newblock In {\em Evolutionary Computation (CEC), 2013 IEEE Congress on},
  1037--1044.

\bibitem[\protect\citeauthoryear{El~Yafrani and
  Ahiod}{2016}]{ElYafrani:2016:PVS:2908812.2908847}
El~Yafrani, M., and Ahiod, B.
\newblock 2016.
\newblock Population-based vs. single-solution heuristics for the travelling
  thief problem.
\newblock In {\em Proceedings of the Genetic and Evolutionary Computation
  Conference 2016}, GECCO '16,  317--324.
\newblock New York, NY, USA: ACM.

\bibitem[\protect\citeauthoryear{Faulkner \bgroup et al\mbox.\egroup
  }{2015}]{Faulkner2015}
Faulkner, H.; Polyakovskiy, S.; Schultz, T.; and Wagner, M.
\newblock 2015.
\newblock Approximate approaches to the traveling thief problem.
\newblock In {\em Proceedings of the 2015 Annual Conference on Genetic and
  Evolutionary Computation}, GECCO '15,  385--392.
\newblock New York, NY, USA: ACM.

\bibitem[\protect\citeauthoryear{{GOODYEAR}}{2008}]{GOODYEAR}
{GOODYEAR}.
\newblock 2008.
\newblock {Factors Affecting Truck Fuel Economy}.
\newblock
  http://www.goodyeartrucktires.com/pdf/re\-sour\-ces/pu\-bli\-ca\-tions/Fac\-tors
  Af\-fec\-ting Truck Fuel Eco\-no\-my.pdf.

\bibitem[\protect\citeauthoryear{Hochbaum}{1997}]{HochbaumApproximation}
Hochbaum, D.
\newblock 1997.
\newblock {\em Appromixation Algorithms for NP-hard Problems}.
\newblock PWS Publishing Company.

\bibitem[\protect\citeauthoryear{Hoy and
  Nikolova}{2015}]{DBLP:conf/aaai/HoyN15}
Hoy, D., and Nikolova, E.
\newblock 2015.
\newblock Approximately optimal risk-averse routing policies via adaptive
  discretization.
\newblock In Bonet, B., and Koenig, S., eds., {\em Proceedings of the
  Twenty-Ninth {AAAI} Conference on Artificial Intelligence, January 25-30,
  2015, Austin, Texas, {USA.}},  3533--3539.
\newblock {AAAI} Press.

\bibitem[\protect\citeauthoryear{Lin \bgroup et al\mbox.\egroup }{2014}]{Lin14}
Lin, C.; Choy, K.; Ho, G.; Chung, S.; and Lam, H.
\newblock 2014.
\newblock Survey of green vehicle routing problem: Past and future trends.
\newblock {\em Expert Systems with Applications} 41(4, Part 1):1118 -- 1138.

\bibitem[\protect\citeauthoryear{Mei, Li, and
  Yao}{2016}]{DBLP:journals/soco/MeiLY16}
Mei, Y.; Li, X.; and Yao, X.
\newblock 2016.
\newblock On investigation of interdependence between sub-problems of the
  travelling thief problem.
\newblock {\em Soft Comput.} 20(1):157--172.

\bibitem[\protect\citeauthoryear{Pisinger}{2005}]{Pisinger20052271}
Pisinger, D.
\newblock 2005.
\newblock Where are the hard knapsack problems?
\newblock {\em Computers \& Operations Research} 32(9):2271 -- 2284.

\bibitem[\protect\citeauthoryear{Polyakovskiy and Neumann}{2015}]{sergey15}
Polyakovskiy, S., and Neumann, F.
\newblock 2015.
\newblock Packing while traveling: Mixed integer programming for a class of
  nonlinear knapsack problems.
\newblock In Michel, L., ed., {\em Integration of AI and OR Techniques in
  Constraint Programming}, volume 9075 of {\em Lecture Notes in Computer
  Science}. Springer International Publishing.
\newblock  332--346.

\bibitem[\protect\citeauthoryear{Polyakovskiy and
  Neumann}{2016}]{DBLP:journals/corr/PolyakovskiyN15}
Polyakovskiy, S., and Neumann, F.
\newblock 2016.
\newblock The packing while traveling problem.
\newblock {\em European Journal of Operational Research} ~--.
\newblock http://dx.doi.org/10.1016/j.ejor.2016.09.035 (in press).

\bibitem[\protect\citeauthoryear{Polyakovskiy \bgroup et al\mbox.\egroup
  }{2014}]{Polyakovskiy2014TTP}
Polyakovskiy, S.; Bonyadi, M.~R.; Wagner, M.; Michalewicz, Z.; and Neumann, F.
\newblock 2014.
\newblock A comprehensive benchmark set and heuristics for the traveling thief
  problem.
\newblock In {\em Proceedings of the 2014 Annual Conference on Genetic and
  Evolutionary Computation}, GECCO '14,  477--484.
\newblock New York, NY, USA: ACM.

\bibitem[\protect\citeauthoryear{Toth}{1980}]{Toth1980}
Toth, P.
\newblock 1980.
\newblock Dynamic programming algorithms for the zero-one knapsack problem.
\newblock {\em Computing} 25(1):29--45.

\bibitem[\protect\citeauthoryear{Yang and
  Nikolova}{2016}]{DBLP:conf/aaai/YangN16}
Yang, G., and Nikolova, E.
\newblock 2016.
\newblock Approximation algorithms for route planning with nonlinear
  objectives.
\newblock In Schuurmans, D., and Wellman, M.~P., eds., {\em Proceedings of the
  Thirtieth {AAAI} Conference on Artificial Intelligence, February 12-17, 2016,
  Phoenix, Arizona, {USA.}},  3209--3217.
\newblock {AAAI} Press.

\end{thebibliography}
\end{document}